\documentclass{article}

\usepackage[nonatbib]{neurips_2020}

\usepackage{calc,amsfonts,amssymb,amsmath,bm,url,color,graphicx,cite,epstopdf,nicefrac,bbold, amsthm}
\usepackage{psfrag,subfigure,float}
\usepackage{multirow}
\usepackage{algorithm}
\usepackage{algpseudocode}

\usepackage[normalem]{ulem}
\usepackage{stmaryrd}
\usepackage{xcolor}
\usepackage{psfrag,framed,mdframed}
\usepackage{comment}
\usepackage{mdframed}
\usepackage{cleveref}
\usepackage{bibunits}
\usepackage[utf8]{inputenc}

\usepackage{mathtools}

\definecolor{orange}{RGB}{255,107,0}
\definecolor{green}{RGB}{0,100,0}

\usepackage{steinmetz}
\renewcommand{\a}{\boldsymbol{a}}
\renewcommand{\b}{\boldsymbol{b}}

\renewcommand{\c}{\boldsymbol{c}}

\newcommand{\h}{\boldsymbol{h}}

\newcommand{\p}{\boldsymbol{p}}

\newcommand{\w}{\boldsymbol{w}}
\newcommand{\y}{\boldsymbol{y}}
\newcommand{\x}{\boldsymbol{x}}
\newcommand{\z}{\boldsymbol{z}}

\newcommand{\zero}{\boldsymbol{0}}
\newcommand{\one}{\boldsymbol{1}}
\newcommand{\I}{\boldsymbol{I}}
\newcommand{\D}{\boldsymbol{D}}

\newcommand{\Q}{\boldsymbol{Q}}
\newcommand{\X}{\boldsymbol{X}}

\renewcommand{\H}{\boldsymbol{H}}

\newcommand{\A}{\boldsymbol{A}}

\newcommand{\V}{\boldsymbol{V}}

\newcommand{\T}{{\!\top\!}}

\newcommand{\bbR}{\mathbb{R}}
\newcommand{\bbQ}{\mathbb{Q}}
\newcommand{\bbC}{\mathbb{C}}

\DeclareMathOperator*{\minimize}{\textrm{minimize}}

\newtheorem{proposition}{Proposition}

\usepackage{xcolor}
\usepackage{psfrag,framed}
\usepackage{lipsum}
\usepackage{chapterbib}
\PassOptionsToPackage{normalem}{ulem}
\usepackage{ulem}
\definecolor{shadecolor}{RGB}{220,220,220}

\title{NP-Completeness of Multicast Beamforming in Wireless Communication}

\author{%
  Sagar Shrestha \\
  Department of Electrical Engineering and Computer Science\\
  Oregon State University\\
  Corvallis, OR 97331 \\
  \texttt{shressag@oregonstate.edu} \\
}

\begin{document}

\maketitle

\begin{abstract}
In this work, the classical problem of multi-cast beamforming in wireless communication is reconsidered. \cite{sidiropoulos2006transmit} showed that the multi-cast beamforming problem is NP-hard. In this project, we show that the corresponding decision problem for real channel matrices and beamformers is NP-complete. Finally, we carry out simulations to reveal the computational complexity of solving the NP-complete beamforming problem in various setting using SAT/SMT solvers. 
\end{abstract}

\section{Introduction}
    In wireless communication, beamforming refers to a technique by which a transmitter with multiple antennas can focus wireless signal towards a receiver. Multi-cast beamforming refers to focusing a common wireless signal towards multiple recievers that may be spatially distributed. This technique is crucial for efficient audio and video streaming in the multiuser setting \cite{konar2016fast}.  
    In this work, we focus on single group multicast beamforming, which is of great interest in the wireless communication \cite{lu2017efficient, konar2016fast}. 
    
    Single group multicast beamforming was shown to be NP-hard in \cite{sidiropoulos2006transmit}. Specifically, \cite{sidiropoulos2006transmit} showed that any instance of the partition problem \cite{garey1979computers}, which is known to be NP-complete, can be reduced to a special instance of the single group multicast beamforming problem. This showed that such special instance can be reduced to any NP-complete problem, as there exist a one to one mapping with a NP-complete problem. However, the special instance used in \cite{sidiropoulos2006transmit} does not cover any cases interesting or realistic in wireless communication. In this work, we show that a broader subset of the beamforming problem is NP-complete, mainly the case where the entries of channel matrix and beamformer are real numbers instead of complex numbers. Although in practice, the channel matrix are almost always complex, real channel matrix can approximate channel conditions with negligible phase shifts. Similarly, real precoder/beamformer vectors would correspond to discarding quadrature component of modulation. 
    
    Finally, we analyse the size of Boolean formula that will be required to encode the beamforming problem, and offer synthetic data simulations using SMT solvers to find solutions of the considered NP-complete problem under various settings.
    
    \textbf{Notation.} $x$, $\x$, and $\X$ denote a scalar, a vector, and a matrix, respectively. $\x_n$ denotes the $n$th column of matrix $\X$. $x(n)$ denotes the $n$th element of the vector $\x$. $\|\cdot\|_2$ denotes the vector $\ell_2$ norm. $\I_N$, $\one_N$, and $\zero_N$ denote an identity matrix of size $N$, $N$ dimensional vector of ones, and $N$ dimensional vector of zeros. ${\rm Diag}(\cdot)$ returns a diagonal matrix with the input vector on the diagonal. $\X^\T$ and $\X^H$ denote the transpose and hermitian of $\X$.
    
\section{Background}
    In a wireless communication setting, the transmitted signal suffers attenuation or even amplification, along with phase shift due to environmental factors. A channel matrix is used to model the such effect. To that end, we use the following signal model signal model which is experimentally verified and physics backed \cite{goldsmith2005wireless}. Let $s \in \bbC$ denote the common signal to be transmitted from a transmitter with $N$ antennas to $M$ receivers with single antenna each. Suppose that $\H = [\h_1, \dots, \h_M] \in \bbC^{N \times M}$ denote the channel matrix where $\h_m$ is the wireless channel from the transmitter to the $m$th user. Then the received signal $y_m \in \bbC^M$ by the $m$th user can be written as:
    $$ y_m = \h_m^H \w s + \eta_m,$$
    where $\eta_m \in \bbC$ is a zero mean circularly symmetric random Gaussian noise of variance $\sigma^2$ and $\w \in \bbC^N$ is the precoding/beamforming vector computed by the transmitter. The problem of beamforming is to design the beamformer $\w$ such that some utility function is maximized. A popular and useful utility function studied widely in the literature is the \textit{quality of service} (QoS) formulation. Under the QoS formulation, the transmitter tries to meet a specified \textit{signal to noise ratio} SNR at the receiving end, while using as small power as possible at the transmitter. This can be formulated as follows:
    \begin{align}\label{eq:beamforming}
        \minimize_{\w \in \bbC^N} ~ & \|\w\|_2 \\
        \text{subject to} ~ & \frac{|\h_m^H \w|}{\sigma^2} \geq \gamma_m, \forall m \in [M], \nonumber
    \end{align} 
    where $\frac{|\h_m^H \w|}{\sigma^2}$ is the SNR achieved at receiver $m$, and $\gamma_m \in \bbR$ is the minimum specified SNR for receiver $m$. 
    Problem \eqref{eq:beamforming} can be re-written as 
    \begin{align}
        \minimize_{\w \in \bbC^N} ~ & \|\w\|_2 \\
        \text{subject to} ~ & {|\widetilde{\h}_m^H \w|} \geq 1, \forall m \in [M], \nonumber
    \end{align} 
    where $\widetilde{\h}_m = \frac{\h_m}{\gamma_m \sigma_m^2}$. Hence without loss of generality, we assume $\sigma_m^2 \gamma_m = 1, \forall m \in [M]$. Problem \eqref{eq:beamforming} was shown to be NP-hard in \cite{sidiropoulos2006transmit}. In this report, we show that when $\w$ and $\H$ are constrained to be real, the resulting beamforming problem is NP-complete. Therefore, we can use SMT solvers to find solutions to the problem under various setting. 

\section{NP completeness of the beamforming problem}
    Since irrational numbers cannot be represented by a finite size symbols, we restrict $\H \in \bbQ^{N \times M}$, i.e., rational numbers. Then the beamforming problem can be written as
    \begin{align}\label{eq:beamforming_real}
        \minimize_{\w \in \bbR^N} ~ & \|\w\|_2 \\
        \text{subject to} ~ & |\h_m^\T \w| \geq 1, \quad \forall m \in [M], \nonumber
    \end{align} 
    
    Problem \eqref{eq:beamforming_real} is not a decision problem. We first consider the corresponding decision problem where we ask the question whether $\|\w\|_2 \leq \kappa, \kappa \in \bbQ_+$  exists such that $|\h_m^\T \w| \geq 1, \forall m \in [M]$. Let ${\rm BF}$ denote the language that consists of all yes instances of such decision problem, i.e., 
    $$ {\rm BF} = \{ (\H, \kappa) ~|~ \exists \w : \|\w\|_2 \leq \kappa,  |\h_m^\T \w| \geq 1, \forall m \in [M]\},$$
    At first glance, it appears obvious that ${\rm BF} \in {\rm NP}$ by taking $\w$ to be the witness, and the conditions can be checked in polynomial time. However, for $\w$ to be a valid witness, we have to show that the number of digits in $\w$ is polynomial in that of the input $(\H, \kappa)$. To clarify, if the solution of \eqref{eq:beamforming_real} is $\w^\star$ and $\|\w^\star\|_2 \leq \kappa$, then there exists a witness $\w$ such that the following holds:
    \begin{itemize}
        \item $\w$ satisfies the constraint in \eqref{eq:beamforming_real}, and $\|\w\|_2 \leq \kappa$,
        \item $\w$ is a rational number with number of digits polynomial in size of the total number of digits in $(\H, \kappa)$ 
    \end{itemize}
    
    Next, consider the following proposition:
    \begin{proposition}
    BF is NP-complete.
    \end{proposition}
    \begin{proof}
        To prove the NP-completeness of BF, we first use the fact that quadratic programs are NP-complete \cite{vavasis1990quadratic, garey1979computers}, and show that BF is in NP. Then we use the idea from \cite{sidiropoulos2006transmit} to establish that BF is NP-hard. Combining the two results, we conclude that BF is NP-complete.
        
        \subsection{${\rm BF} \in {\rm NP}$}
        First, we show that BF is in NP by reducing BF to a quadratic program. We represent the decision problem associated with quadratic program as QP, which is defined as follows:
        $$ {\rm QP} = \{ (\A, \b, \c, \D, \kappa) ~|~ \exists \w : \w^\T \A \w + \c^\T \w \leq \kappa,  \D \w \geq \b \}$$
        where $\A \in \mathbb{S}^{P \times P}, \D \in \mathbb{Q}^{R \times P}, \b \in \mathbb{Q}^{R}, \c \in \mathbb{Q}^P, \kappa \in \bbQ_+$. 
        \cite{vavasis1990quadratic} showed that ${\rm QP} \in {\rm NP}$. This implies that there exists a polynomial time TM ${\rm T}_{\rm QP}$ such that:
        $$ {\rm QP} = \{ (\A, \b, \c, \D, \kappa) ~|~ \exists \y : {\rm T}_{\rm QP}((\A, \b, \c, \D, \kappa), \y) = 1\}.$$
        where the witness $\y$ is polynomial in the size of the input.
        
        Let use define a variable $\z \in \{+1, -1\}^M$. Define $\widetilde{\A}, \widetilde{\b}, \widetilde{\c}, \widetilde{\D}$ as follows:
        \begin{align*}
            \widetilde{\A} &= \I_N \\
            \widetilde{\D}(\z) & = {\rm Diag}(\z) \H^\T \\
            \widetilde{\b} &= \one_M \\
            \widetilde{\c} &= \zero_{N}. 
        \end{align*}
        Then one can write rewrite problem \eqref{eq:beamforming_real} as follows:
        \begin{align}\label{eq:bf_modified}
            \minimize_{\w \in \mathbb{Q}^N, \z \in \{+1, -1\}^M} ~& \w^\T \widetilde{\A} \w + \widetilde{\c}^\T \w\nonumber \\
            \text{subject to} \quad \quad ~& \widetilde{\D}(\z) \w \geq \widetilde{\b}
        \end{align}
        It is straightforward to see that the objective of \eqref{eq:beamforming_real} is equivalent to that of \eqref{eq:bf_modified}. To see how the constraints are equivalent, observe that the $m$th constraint of \eqref{eq:bf_modified} is
        \begin{align*}
            z(m) \h_m^\T \w  & \geq 1.
        \end{align*}
        Since $z(m)$ can take the value $\{+1, -1\}$, the above inequality corresponds to
        $$ \text{either } \h_m^\T \w_N \geq 1 \text{ or } - \h_m^\T \w_N \geq 1,$$
        which is equivalent to $|\h_m^\T \w_N| \geq 1$.
        
        Note that problem \eqref{eq:bf_modified} is not a quadratic program because $\widetilde{\D}(\z)$ is not a constant. The constraints are in fact non-convex, whereas for a quadratic program the constraints need to be affine. However, we can use existential quantifier to write the BF problem as follows:
         $$ {\rm BF} = \{ (\H, \kappa) ~|~ \exists \z \exists \y : {\rm T}_{\rm QP}((\widetilde{\A}, \widetilde{\b}, \widetilde{\c}, \widetilde{\D}(\z), \kappa), \y) = 1\}.$$
        Hence the concatenated string $\z\y$ take the role of a witness. Since, we know that for any value of $\z \in \{ -1, 1\}^N$, $(\widetilde{\A}, \widetilde{\b}, \widetilde{\c}, \widetilde{\D}(\z), \kappa)$ defines a quadratic program, the witness of QP, $\y$, is polynomial in size of the input due to \cite{vavasis1990quadratic}. Finally, as $\z$ can be represented by $N$ bits, it is polynomial in size of the input. Hence $\z\y$ is polynomial in size of the input.
        
        Thus, write 
        $$ {\rm BF} = \{ (\H, \kappa) ~|~ \exists \z\y : {\rm T}_{\rm BF}((\H, \kappa), \z\y) = 1\} $$
        for a polynomial time TM ${\rm T}_{\rm BF}$ detailed in Algorithm \ref{algo:t_bf}, where $\y\z$ is polynomial in the size of the input. 

        \begin{algorithm}
        \caption{${\rm T}_{\rm BF} ((\H, \kappa), \y \z)$}\label{alg:t_bf}
            \begin{algorithmic}
            \State $\widetilde{\A} = \I_N$
            \State $\widetilde{\D}  = {\rm Diag}(\z) \H^\T $
            \State $\widetilde{\b} = \one_M $
            \State $\widetilde{\c} = \zero_{N}$
            \State return ${\rm T}_{\rm QP}( (\widetilde{\A}, \widetilde{\b}, \widetilde{\c}, \widetilde{\D}, \kappa), \y )$
            \end{algorithmic}
        \end{algorithm}
        Finally, since the construction of the variables $\widetilde{\A}, \widetilde{\b}, \widetilde{\c}, \widetilde{\D}(\z)$ can be done in polynomial time, and ${\rm T}_{\rm QP}$ is a polynomial time TM, ${\rm T}_{\rm BF}$ is a polynomial time TM. Hence ${\rm BF} \in {\rm NP}$.
        
        \subsection{${\rm PARTITION} \leq_{p} {\rm BF}$}
        ${\rm PARTITION}$\cite{garey1979computers} can be Karp reduced to ${\rm BF}$ by using ideas from \cite{sidiropoulos2006transmit}.
        
        First, the problem PARTITION asks whether a list of integers $a_1, \dots, a_N$ can be partitioned into two sets such that the sum of the two sets are equal. In other words, we search for variables $y_1, \dots, y_N \in \{ +1, -1\}^N$ such that
        $$ \sum_{n=1}^N y_n a_n = 0. $$
        
        Let $\p = [p_1, \dots, p_N]^\T \in \bbR^{N}$. The PARTITION problem is equivalent to asking whether the following problem as a solution of $N$:
        \begin{align}\label{eq:partition_bf}
            \minimize_{\p \in \bbR^{N}} ~& p_1^2 + p_2^2 + \dots + p_N^2 + \left(\sum_{n=1}^N p_n a_n\right)^2 \\
            \text{subject to } ~& p_n^2 \geq 1, \quad n \in [N] \nonumber
        \end{align}
        If the above problem has solution $N$, then $p_n \in \{+1, -1\}$ and $\sum_{n=1}^N p_n a_n = 0$, i.e., $a_1, \dots, a_n$ can be partitioned. Therefore, there is one to one correspondence of the above problem and PARTITION. However, the above problem is a special instance of BF. To see this, we can write the objective of \eqref{eq:partition_bf} as $\p^\T \Q \p$, where $\Q = \I_N + \a \a^\T$, where $\a = [a_1, \dots, a_N]^\T$. Since $\Q$ is a positive definite matrix, we can write $\Q = \V^\T \V$, for some full rank $\V$. Then we make a change of variable to $\widetilde{\w} = \V \y$. This implies that $\y = \V^{-1} \w$. We represent $\widetilde{\h}_n = \V^{-1}(n,:)^\T$. Hence we can equivalently write \eqref{eq:partition_bf} as follows:
        \begin{align}
            \minimize_{\widetilde{\w}  \in \bbR^{n}} &~ \| \widetilde{\w} \|_2^2 \\
            \text{subject to } &~ |\widetilde{\h}_n| \geq 1, \quad n \in [N]. \nonumber
        \end{align}
        Therefore, $\a \in {\rm PARTITION}$ if and only if $((\V^{-1})^\T, N) \in BF$, where $\V$ is obtained by square root decomposition of $\I_N + \a \a^\T$. Since the square root decomposition via SVD, matrix inversion and multiplication can be done in polynomial time, ${\rm PARTITION} \leq_p {\rm BF}$. Hence BF is NP-hard
        
        Since ${\rm BF} \in {\rm NP}$ and ${\rm BF}$ is NP-hard, BF is NP-complete.
    \end{proof}

    \section{Encoding an instance of BF as a Boolean formula}
    In the previous section, we showed that BF is NP-complete. In order to encode an instance of BF as a Boolean formula, we restrict our attention to a fixed size representation (e.g., 64 bits) for each rational numbers, i.e., the input and the witness for ${\rm BF}$. The reason is that if we merely assume that the witness is polynomial in size with respect to the input string, it is difficult to quantify exactly the number of variables and clauses. Secondly, the SMT solver that we will be using also represents the real numbers using a fixed size. 
    
    Let $Q$ be the number of bits used to represent a single scalar. Since the number of scalar variables is $N$ and Boolean variables is $M$, we have $NQ + M$ Boolean variables of interest in total. Note, however, that the total number of Boolean variables in the final formula can be far greater than this, mainly due to the need for representing intermediate results as variables (e.g., each constraint will be converted into a Boolean formula with its own set of variables, and later, variables from all constraints are enforced to be equivalent).
    
    In order to solve problem BF in practice, we will provide the following linear and quadratic inequalities to the SMT solver. 
    \begin{subequations}
    \begin{align}
        & w(1)^2 + w(2)^2 + \dots + w(N)^2 \leq \kappa \label{eq:non_linear_constraint} \\
        & z(1)(w(1)h_m(1) + w(2)h_m(2) + \dots + w(N)h_m(N)) \geq 1 , \forall m \in [M]\label{eq:linear_constraint}
    \end{align}
    \end{subequations}
    The solver attempts to find a solution $\w, \z$ that satisfy the above inequalities. In order to convert the above inequalities into a Boolean formula needed to represent the above inequalities, the following steps are followed:
    \begin{enumerate}
        \item obtain Boolean circuit blocks for summation, multiplication, and comparison, 
        \item compose the circuit blocks in order to obtain individual inequalities
        \item convert the resulting Boolean circuit into Boolean formula, and
        \item add clauses for ensuring variable consistency
        \item count the size of the final Boolean formula. 
    \end{enumerate}
    
    Let $C_{\rm sum}, C_{\rm mult}$, and $C_{\rm comp}$ denote the size of the Boolean formula corresponding to the Boolean circuits for the summation, multiplication, and comparison of two scalar values represented by $Q$ bits. For \eqref{eq:non_linear_constraint}, we require $N$ multiplication blocks, $N-1$ summation blocks, and 1 comparision block. For one inequality of \eqref{eq:linear_constraint}, we require $N$ multiplication blocks, $N-1$ summation blocks, 1 comparison block and one Boolean circuit for changing the sign of $\w_m^T \h$ according to the value of $z_m$. Let $C_{\rm sign}$ be the size of the Boolean formula associated with the Boolean circuit for changing the sign of a scalar represented by $Q$ bits according to the value of another Boolean variable. In total, \eqref{eq:linear_constraint} requires $NM$ multiplications, $(N-1)M$ summations, $M$ sign determination, and $M$ comparisons. 
    
    Finally, we need to ensure that the Boolean variables in the Boolean formula for each of the inequalities are consistent with each other. This can be done by adding extra clauses. For example, if $x^{(1)}, x^{(2)}, x^{(3)}$ denote the variables used in the 3 inequalities to represent the same variable $x$, then consistency between them can be enforced by adding the following clause
    $$\left(\left( x^{(1)} \land x^{(2)} \land x^{(3)} \right) \lor \left( \overline{x}^{(1)} \land \overline{x}^{(2)} \land \overline{x}^{(3)} \right) \right).$$
    Let $C_{\rm consist}$ denote the size of the boolean formula required to ensure consistency between two scalars taking $Q$ bits each. The number of consistency checks required is $M-1$. 
    
    From the above discussion, the total size of the boolean formula, $C_{\rm total}$ is approximately
    $$ C_{\rm total} \approx (NM + N)C_{\rm mult} + (N-1)(M+1) C_{\rm sum} + (M+1) C_{\rm comp} + M C_{\rm sign} + (M-1)C_{\rm consist}.$$
    The above is an approximation because we have ignored the increase in size due to concatenation of different inequalities, and the additional consistency clauses that may be required to compose different blocks.
    
    If we fix one of $M$ and $N$ constant, we can see that the size of the Boolean formula grows linearly with respect to the other parameter. 
    
    \section{Simulations}
    In this section, we run simulations to observe the performance of problem BF for various problem setting. We use SMT solver, Z3 \cite{moura2008z3}, to solve problem BF. Specifically, for different realizations of $\H$ and $\kappa$, we observe whether $(\H, \kappa) \in {\rm BF}$. Further, we observe relation between the various problem parameters and the run-time of the SMT solver. The code used in the following simulations is available online\footnote{https://github.com/shresthasagar/beamforming-np-completeness-toc.git}.

    \textbf{Optimal Solution}:
    In order to verify the correctness of the implementation, we obtain an optimal solution to the given problem with enumeration. The enumeration method works as follows. If we fix $\z$ to be $\widehat{\z} \in \{-1, +1\}^N$, problem \eqref{eq:bf_modified} is a quadratic program, which is a convex program. Hence any convex programming solver can be used to obtain a solution arbitrarily close to the optimal. Therefore, we can evaluate problem \eqref{eq:bf_modified} for all $2^M$ possible values of $\z$ and select the solution $\x^\star$ and $\z^\star$ that attains the minimum objective value $v^\star$. 
    
    \textbf{Data Generation}:
    $M$ and $N$ can be varied arbitrarily. Note that the problem is feasible for any value of $M,N$. All entries of $\H$ are sampled independently from a zero mean Gaussian distribution with variance 1. This is consistent with Rayleigh fading channel \cite{goldsmith2005wireless}, which is a widely adopted channel model in wireless communications research. In Rayleigh fading channel the entries of $\H$ is sampled from circularly symmetric Gaussian distribution. In the following experiments, $\sigma_m^2$ and $\gamma_m$ are set to $1$ for all $m \in [M]$.
    
    \subsection{Sanity Check}
    First, we verify the solution provided by the SMT solver by comparing it to the optimal solution. Table \ref{tab:sanity_check} shows the solution obtained for a particular realization of $\H$ with $(N,M) = (2,3)$ while $\kappa$ is varied relative to the optimal solution $v^\star$ obtained by using the enumeration procedure. One can see that when $(\H, \kappa) \not \in {\rm BF}$, the solver takes much longer time to return failure. when $\kappa \geq v^\star$, the solver returns a solution and successful set membership status in relatively small time.  
    
    \begin{table}[h]
        \centering
        \caption{Set membership, objective value and time for various $\kappa$ relative to the optimal objective value $v^\star$.}
        \begin{tabular}{|c|c|c|c|}
            \hline
             $\kappa$ &  $(\bm H, \kappa) \in {\rm BF}$ & Objective & Time(seconds) \\ \hline
             $0.5 v^\star$ & False & - & 0.29 \\ \hline
             $0.99 v^\star$ & False & - & 72.51 \\ \hline
             $v^\star$ & True & $v^\star$ & 0.05 \\ \hline
             $1.1 v^\star$ & True & $1.07 v^\star$ & 0.08 \\ \hline
             $1.5 v^\star$ & True & $1.31 v^\star$ & 0.06 \\ \hline
             $2 v^\star$ & True & $1.23 v^\star$ & 0.03 \\ \hline
        \end{tabular}
        \label{tab:sanity_check}
    \end{table}
    
    \subsection{Further Experiments}
    Now, we observe solver time for various sizes of problem parameters, i.e., the number of antennas and users. Figure \ref{fig:varying_params} shows the time taken by the solver for various parameter settings. The results are averaged over $15$ randomly sampled $\H$.

    \begin{figure}[h]
        \centering
        \includegraphics[width=\linewidth]{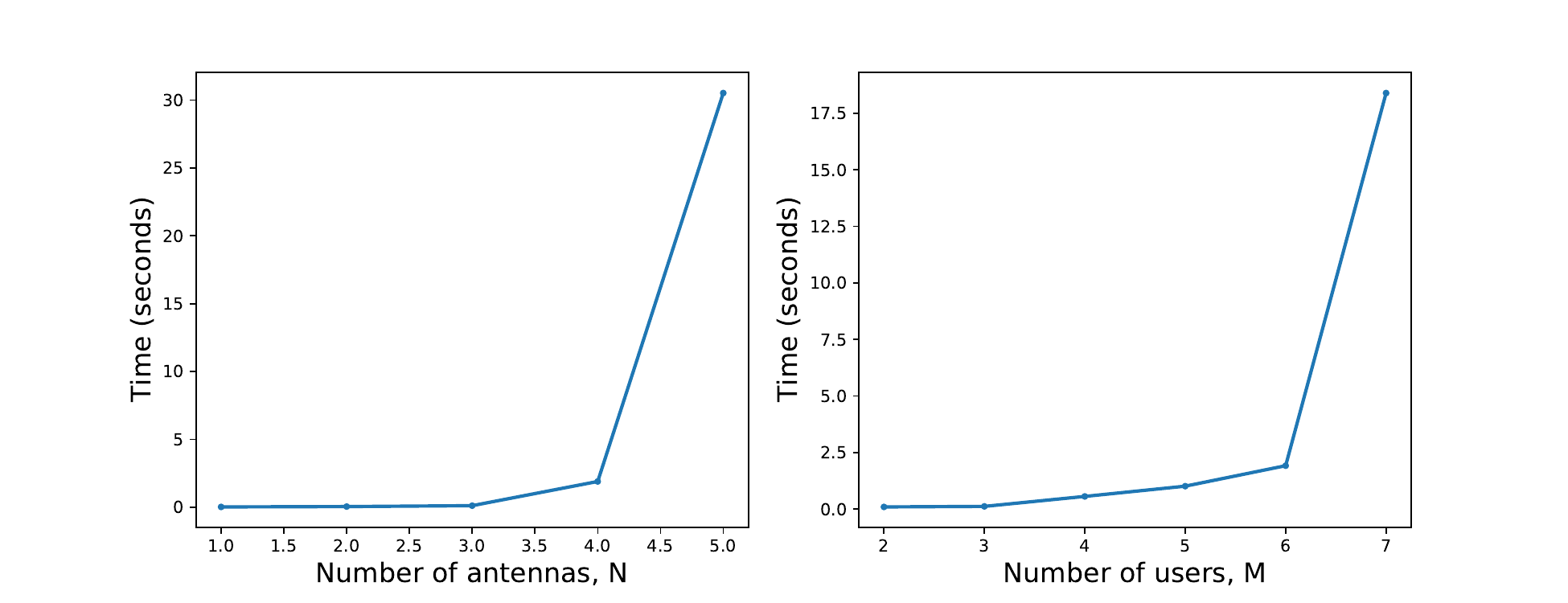}
        \caption{[Left] Time taken by the solver for various number of antennas, when $M=4$. [Right] Time taken by solver for various number users, when $N=4$.}
        \label{fig:varying_params}
    \end{figure}
    
    In figure~\ref{fig:varying_params}[Left] shows the time taken by the solver for varying number of antennas when the number of users is fixed to $4$. Small scale scenarios such as with WiFi can be reflected by this problem size. However, large scale scenarios (e.g., a 5G base station serving many users can contain 10s of antennas) can not be covered. Here, we have set $\kappa$ to $100 v^\star$. We can see that the time taken by the solver grows exponentially with the number of antennas. This is expected as the solving time for Boolean formula grows exponentially with the number of variables.  
    
    Figure~\ref{fig:varying_params}[Right] shows the time taken by the solver for varying number of users when the number of antennas is fixed to $4$. We can see that, similar to the previous case, the solving time increases exponentially with the number of users. However, time consumed seems to grow much faster with increase in $N$ than $M$. The reason could be that the size of the constraint $\|\w\|_2^2 \leq \kappa$ grows only with the increase in $N$.
    
    \section{Conclusion}
    In this project, we studied single group multicast beamforming problem in wireless communication in the case where both channel matrix and beamforming vector is real. We showed that, in such a case, the resulting decision problem is NP-complete. Using this knowledge, we first analysed the size of the Boolean formula needed to represent the problem using finite precision for each scalar. Finally, we observed in simulations the dependence of the running time on various parameters of the beamforming problem.
    
    We should remark that experiments with larger problem size compared to the ones shown in Figure \ref{fig:varying_params} seem to be infeasible (some problem instances take more than 5 minutes to solve) in general. 
    
 \bibliographystyle{IEEEtran}
 \bibliography{main}

\begin{thebibliography}{1}
\providecommand{\url}[1]{#1}
\csname url@samestyle\endcsname
\providecommand{\newblock}{\relax}
\providecommand{\bibinfo}[2]{#2}
\providecommand{\BIBentrySTDinterwordspacing}{\spaceskip=0pt\relax}
\providecommand{\BIBentryALTinterwordstretchfactor}{4}
\providecommand{\BIBentryALTinterwordspacing}{\spaceskip=\fontdimen2\font plus
\BIBentryALTinterwordstretchfactor\fontdimen3\font minus
  \fontdimen4\font\relax}
\providecommand{\BIBforeignlanguage}[2]{{%
\expandafter\ifx\csname l@#1\endcsname\relax
\typeout{** WARNING: IEEEtran.bst: No hyphenation pattern has been}%
\typeout{** loaded for the language `#1'. Using the pattern for}%
\typeout{** the default language instead.}%
\else
\language=\csname l@#1\endcsname
\fi
#2}}
\providecommand{\BIBdecl}{\relax}
\BIBdecl

\bibitem{sidiropoulos2006transmit}
N.~D. Sidiropoulos, T.~N. Davidson, and Z.-Q. Luo, ``Transmit beamforming for
  physical-layer multicasting,'' \emph{IEEE transactions on signal processing},
  vol.~54, no.~6, pp. 2239--2251, 2006.

\bibitem{konar2016fast}
A.~Konar and N.~D. Sidiropoulos, ``A fast approximation algorithm for
  single-group multicast beamforming with large antenna arrays,'' in \emph{2016
  IEEE 17th International Workshop on Signal Processing Advances in Wireless
  Communications (SPAWC)}.\hskip 1em plus 0.5em minus 0.4em\relax IEEE, 2016,
  pp. 1--5.

\bibitem{lu2017efficient}
C.~Lu and Y.-F. Liu, ``An efficient global algorithm for single-group multicast
  beamforming,'' \emph{IEEE Transactions on Signal Processing}, vol.~65,
  no.~14, pp. 3761--3774, 2017.

\bibitem{garey1979computers}
M.~R. Garey and D.~S. Johnson, \emph{Computers and intractability}.\hskip 1em
  plus 0.5em minus 0.4em\relax freeman San Francisco, 1979, vol. 174.

\bibitem{goldsmith2005wireless}
A.~Goldsmith, \emph{Wireless communications}.\hskip 1em plus 0.5em minus
  0.4em\relax Cambridge university press, 2005.

\bibitem{vavasis1990quadratic}
S.~A. Vavasis, ``Quadratic programming is in np,'' \emph{Information Processing
  Letters}, vol.~36, no.~2, pp. 73--77, 1990.

\bibitem{moura2008z3}
L.~d. Moura and N.~Bj{\o}rner, ``Z3: An efficient smt solver,'' in
  \emph{International conference on Tools and Algorithms for the Construction
  and Analysis of Systems}.\hskip 1em plus 0.5em minus 0.4em\relax Springer,
  2008, pp. 337--340.

\end{thebibliography}

\end{document}